\documentclass[12 pt]{article}
\usepackage[utf8]{inputenc}
\usepackage{amsmath}
\usepackage{amssymb}
\usepackage{bbm}
\usepackage{amsfonts}
\usepackage{colortbl}
\usepackage[dvipsnames]{xcolor}
\usepackage{mathdots}
\usepackage[margin=1.25in]{geometry}
\usepackage{float}
\usepackage{multirow,bigdelim}
\usepackage[ruled,vlined]{algorithm2e}
\usepackage{xr-hyper}
\usepackage{hyperref,theoremref}
\usepackage[onehalfspacing]{setspace}
\usepackage{pst-node}

\usepackage{tikz-cd} 
\usepackage{amsthm}
\usepackage{enumerate}
\usepackage{thmtools, thm-restate}
\usepackage{xcolor}
\usepackage{titlesec}
\titlespacing{\section}{0pt}{1pt plus 0.5pt minus 0.5pt}{0pt plus 0.5pt minus 0.5pt}
\titlespacing{\subsection}{0pt}{1pt plus 0.5pt minus 0.5pt}{0pt plus 0.5pt minus 0.5pt}

\usepackage{etoolbox}

\pretocmd{\section}{\ifdim\pagetotal>0pt\vspace*{-\parskip}\fi}{}{}
\pretocmd{\subsection}{\ifdim\pagetotal>0pt\vspace*{-\parskip}\fi}{}{}

\newcommand{\strow}[2]{%
  \makebox[1.8em][l]{%
    $\ifdim #2pt<0.327pt \mathord{*}\fi%
      \ifdim #2pt<0.2725pt \kern0.05em\mathord{\dagger}\fi%
      \ifdim #2pt<0.1635pt \kern0.05em\mathord{\bullet}\fi$%
  }%
  #1 & #2\\%
}

\usepackage{caption}
\usepackage{enumitem}
\setlist[enumerate]{topsep=0pt}
\setlist[itemize]{nosep}       
\setlist[enumerate]{nosep}
\setlist[itemize]{topsep=0pt}
\usepackage{graphicx}
\usepackage{adjustbox}
\usepackage{enumitem}
\usepackage{booktabs}

\usepackage{relsize,exscale}
\newcommand*\interior[1]{#1^{\mathsf{o}}}

\usepackage{tikz,lipsum,lmodern}
\usetikzlibrary{arrows}
\usetikzlibrary{shapes,decorations,arrows,calc,arrows.meta,fit,positioning}
\tikzset{
    -Latex,auto,node distance =1 cm and 1 cm,semithick,
    state/.style ={ellipse, draw, minimum width = 0.7 cm},
    point/.style = {circle, draw, inner sep=0.04cm,fill,node contents={}},
    bidirected/.style={Latex-Latex,dashed},
    el/.style = {inner sep=2pt, align=left, sloped}
}

\usepackage[round]{natbib}
\setcitestyle{round}
\usepackage[skip=4pt]{caption}
\usepackage[hang,flushmargin]{footmisc}

\newcommand{\indep}{\mathrel{\text{\scalebox{1.07}{$\perp\mkern-10mu\perp$}}}}

\newtheoremstyle{mystyle}
  {2pt} 
  {-3pt} 
  {\itshape} 
  {} 
  {\bfseries} 
  {.} 
  { } 
  {} 
\theoremstyle{mystyle}
\newtheorem{theorem}{Theorem}[section]
\newtheorem{assumption}{Assumption}

\newtheorem{corollary}{Corollary}

\newtheorem{definition}{Definition}[section]
\newtheorem{example}{Example}

\newtheorem{lemma}{Lemma}

\newtheorem{proposition}{Proposition}[section]
\newtheorem{remark}{Remark}

\usepackage{color}
\definecolor{section_color}{rgb}{0.35,0.0,0}
\definecolor{MyGray}{rgb}{0.96,0.97,0.98}

\makeatletter
\renewenvironment{proof}[1][\proofname]{%
  \par\addvspace{0pt}
  \pushQED{\qed}%
  \normalfont                
  \topsep0pt \partopsep0pt \itemsep0pt \parsep0pt%
  \trivlist
  \item[\hskip\labelsep\itshape #1\@addpunct{.}]
  \ignorespaces
}{%
  \popQED\endtrivlist\@endpefalse
  \addvspace{0pt}
}
\makeatother
\usepackage{hyperref}
\hypersetup{
  colorlinks   = true,   urlcolor     = red,   linkcolor    = blue,   citecolor   = blue }

\providecommand{\keywords}[1]
{
  \small	
  \textbf{Keywords:} #1
}

\providecommand{\JEL}[1]
{
  \small	
  \textbf{JEL codes:} #1
}

\newcommand*{\thisdraft}{This draft: January 2026} 
\newcommand*{\firstdraft}{First draft: May 2021}  

\title{A Simple Measure of Robustness for External Validity under Covariate Shifts}
\author{Pietro Emilio Spini\textsuperscript{1}}
\date{
\thisdraft \\
\firstdraft}

\begin{document}

\setlength{\abovedisplayskip}{0.25cm}
\setlength{\belowdisplayskip}{0.25cm}
\setlength{\abovedisplayshortskip}{0.25cm}
\setlength{\belowdisplayshortskip}{0.25cm}

\begin{singlespace} 
\footnotetext[1]{Email: \href{mailto:pietro.spini@bristol.ac.uk}{pietro.spini@bristol.ac.uk}. University of Bristol, 12 Priory Road, BS8 1TU, UK. \\ I thank Yixiao Sun, Kaspar Wuthrich, James Hamilton, Sukjin Han, Sami Stouli, Stefan Hubner, David Pacini, Gregory Jolivet, Itzik Fadlon, Xavier D'Haultfoeuille, Matt Masten, Adam Rosen, Ashesh Rambachan, Davide Viviano, Michael Pollmann, Kirill Ponomarev, and Rami Tabri for their helpful comments.  Participants at EGSC 2021, EWMES 2021, the Microeconometrics Class of 2022-2023 Conference at Duke and seminar participants at PSE-CREST, University of Exeter, the University of Warwick, UvA, Erasmus University Rotterdam, University of Bristol, University of Manchester, University of Surrey, NYUAD, UCSD, University of Victoria, and the Philadelphia Federal Reserve provided valuable discussion. All remaining errors are mine. A previous version of this paper circulated under the title ``Robustness, Heterogeneous Treatment Effects, and Covariate Shifts.''}

\maketitle
\pagenumbering{gobble}

\begin{abstract}
\noindent This paper studies the robustness of estimated policy effects to changes in the distribution of covariates, a key determinant of the external validity of (quasi)-experimental results. 
I propose a novel robustness metric $\delta^*$ which measures the smallest covariate shift needed to invalidate an empirical claim about the policy effect (e.g., $ATE > 0$). I estimate $\delta^*$ via de-biased GMM, achieving a parametric rate of convergence while accommodating machine-learning estimators of treatment-effect heterogeneity (e.g., LASSO, random forests, neural networks). I develop benchmarking and calibration exercises to interpret the magnitude of $\delta^*$. I illustrate these tools in an application to the Oregon Health Insurance Experiment. Researchers can report $\delta^*$ alongside the point estimate and standard error as a third number gauging external validity under covariate shifts.
\end{abstract}

\keywords{Robustness, 
Heterogeneous Treatment Effects, KL divergence, 
Semiparametric estimation, De-biased GMM, Oregon Health Insurance Experiment} 

\JEL{C14, C18, C44, C51, C54, D81, I13}

\end{singlespace}

\pagebreak

\clearpage
\pagenumbering{arabic} 
\setcounter{page}{1}

\section{Introduction}
Evidence-based policy-making uses experimental and quasi-experimental studies to guide the adoption of policies in various settings.
This approach relies on (quasi)-experimental findings being robust and generalizable beyond the original experiment. In practice, this is not always the case: there are several examples of policies that, when implemented in non-experimental settings, fell short of their own experimental estimates \citep{deaton2010instruments, cartwright2012evidence, williams2020external}. Researchers and policy-makers may want to complement their estimates with a tool quantifying the robustness of their findings for policy adoption beyond the experimental setting. 

In this paper, I introduce a new robustness metric, a scalar $\delta^*$, that quantifies how much the characteristics of policy recipients would need to change to invalidate the (quasi)-experimental findings. The metric captures uncertainty arising from systematic differences in recipients' characteristics across environments.\footnote{Quantifying other sources of systematic uncertainty has been a central theme in the recent econometric literature including \citet{andrews2017measuring} for moment conditions, \citet{altonji2005selection, oster2019unobservable, cinelli2020making} for confounding factors, and the breakdown approaches in \citet{horowitz1995identification, masten2020inference, rambachan2023more}.} This contrasts with uncertainty from sampling variation, summarized by the standard errors that routinely accompany point estimates. As such, $\delta^*$ complements standard errors and can be reported alongside them as a ``third number''. To make its magnitude operational for researchers, I suggest interpretation and calibration exercises, including benchmarking $\delta^*$ against covariate shifts in relevant implementation environments.
\par 

As a motivating example, consider a policy-maker who must decide whether to offer medical insurance coverage to low-income households. The policy-maker has access to the experimental estimates of \citet{finkelstein2012oregon} which suggest that a similar intervention led to higher health-care utilization and reduced financial strain for recipients in Oregon. The target population of insurance recipients could differ from the experimental one in Oregon along important dimensions. Our goal is to quantify how robust the experimental findings would be if relevant characteristics of the recipients are allowed to change. In this paper, I provide a solution to this problem by leveraging the policy effect heterogeneity in the experiment.  \par

When policy effects are heterogeneous across sub-populations with different covariate values, (quasi)-experimental findings are generally not robust to changes in the  covariates' distribution. Small changes in the distribution of the covariates could lead to significant aggregate changes in the policy effects. For example, in the Oregon experiment, subsidized health insurance could benefit sicker patients more than healthier patients. Then, the proportion of recipients with a given pre-existing health status, health habits, and/or co-morbidities may strongly influence the overall effect of the policy. 
Often, these covariates are exclusively collected in the experiment and are not all available in the new policy environment prior to implementation. As a result, the reweighting procedures in \citet{hsu2020counterfactual} and \citet{hartman_2020_generalizability} are often infeasible, since they require the full covariate set in the new environment. 
Moreover, the heterogeneity of policy effects  can be hard to model. While domain knowledge can help select covariates that are predictive of the heterogeneity of policy effects, it typically cannot pin down its specific functional form. Because heterogeneity is the channel linking covariate shifts to the aggregate policy effects, a general approach to robustness must reflect the uncertainty regarding the heterogeneity's functional form. \par

My robustness metric avoids the need to impose a functional form for the policy effect heterogeneity, letting it instead be flexibly estimated. When designing a robustness metric for distributional changes, relying on functional form assumptions carries important implications for what type of shifts the metric can detect.\footnote{Many popular existing approaches to robustness and sensitivity analysis, like \citet{altonji2005selection}, \citet{oster2019unobservable} and \citet{cinelli2020making}, take advantage of specific functional forms.} If the way we measure a shift is misaligned with the heterogeneity model, the resulting measure of robustness may be misleading. For example, suppose that distance between two covariate distributions is measured only by their difference in means. With an unrestricted form for the heterogeneity of policy effects, one can construct a mean-preserving shift that invalidates the policy-maker's claim. If in the Oregon experiment, higher-income recipients have negative effects while lower-income recipients have positive effects, a mean-preserving spread of the income distribution could flip the aggregate effect.
Yet, by construction, this shift would have a distance of zero from the experimental covariates, despite changing the experimental findings. This example motivates a robustness metric that accommodates flexible forms of policy effect heterogeneity, whose functional form is, \textit{ex-ante}, unknown. My metric does so while remaining easy to compute and interpret: a one-number summary of heterogeneity which only depends on (quasi)-experimental data. \par

There is a natural connection between the covariate robustness exercise and the literature on Partial Policy Effects. For example, \citet{rothe2012partial} considers the effect of an (infra)-marginal perturbation of the covariate distribution along a fixed direction on a functional of the unconditional outcome distribution. In contrast, in this paper, the direction of the perturbation is not specified \textit{ex-ante} and may itself be the object of interest as it represents, among all possible shifts invalidating the policy-maker's conclusion, the hardest one to detect. This distinction reflects the different purpose and the complementarity of the two approaches. A specific candidate for the covariate distribution is most useful for decomposition exercises that highlight the contribution of several variables on the unconditional distribution, like in the application in \citet{rothe2012partial}. Conversely, searching within a large space of covariate distributions is useful for the policy-maker evaluating the experimental evidence for policy adoption. \par 

Measuring covariate shifts requires choosing a distance between distributions. In my approach, I adopt Kullback-Leibler divergence (KL distance). It is a popular choice for sensitivity analysis, appearing recently in \citet{christensen2023counterfactual} for moment inequality models, \citet{duchi2021learning} for distributionally robust stochastic optimization, and, in a Bayesian context in \citet{ho2023global}. It has several advantages in our context. First, it is invariant to smooth invertible transformations of the covariates, hence independent of the covariates' units \citep{qiao2010study}. Second, it provides a closed form expression for the proposed global robustness measure, while other popular robustness approaches, like \citet{broderick2020automatic} rely on local approximations. Leveraging the closed form solution, I cast estimation of my robustness metric as a GMM problem depending only on the observed covariate distribution and a functional parameter capturing the heterogeneity of policy effects. \par

Policy-effect heterogeneity can often be sparse: from a rich covariate set, only a few are needed to capture observable effect variation.
With many covariates, it can be hard to select which ones are important \textit{ex-ante}. Machine-learning estimators, like LASSO, random forest, or neural networks, can automatically exploit the sparsity and select the key covariates, avoiding \textit{ad-hoc} procedures. Using machine-learning to estimate policy effect heterogeneity is appealing, but it may result in substantial bias in the estimated robustness metric $\delta^*$, due to regularization and/or model selection.
To accommodate machine-learning methods, I construct a de-biased GMM estimator leveraging the theory in \citet{chernozhukov2020locally} to eliminate the first-order bias from first-step estimators. I show that my metric $\delta^*$ can be consistently estimated at $\sqrt{n}$-rate under mild conditions, letting the researcher flexibly choose among many first-step estimators of policy effect heterogeneity.  \par

I apply my robustness procedure to study the Oregon health insurance experiment, whose findings have informed policy adoption in public health \citep{OregonNYT}. Focusing on health-care utilization and financial strain outcomes, I evaluate the robustness of the policy-effects estimates in \citet{finkelstein2012oregon} to covariate shifts. The recipients of the Oregon lottery are predominantly older, in poorer health, and with a larger proportion of White individuals than the national average \citep{finkelstein2013oregon}. These demographic features invite questions about the robustness of the Oregon experiment's outcomes, especially if they are used to shape policies in other states. Differences in the magnitude and sign of the effects of the Medicaid expansions in Oregon and Massachusetts have motivated \citet{kowalski2023reconciling} to investigate the different populations of beneficiaries in the two states. My robustness exercise is complementary: I compute the smallest covariate shift from the Oregon experiment that eliminates the lottery's positive effects on health-care utilization and financial strain. Among the outcomes considered, outpatient visits are the most robust. \par

This paper is also related to the econometric and statistics literature on robustness and sensitivity analysis developed since \citet{tukey1960survey} and \citet{huber1965robust}. Recently, there are many other important but distinct robustness approaches: geared towards external validity \citet{meager2019understanding}, \citet{gechter2015generalizing}, \citet{gechter2024generalizing}, robustness to dropping a percentage of the sample \citet{broderick2020automatic}, by looking at sub-populations \citet{jeong2020robust}, or with respect to unobservable distributions like in \citet{christensen2023counterfactual}, \citet{armstrong2021sensitivity}, \citet{bonhomme2018minimizing}, and \citet{antoine2020robust}, \citet{adjaho2022externally} in the context of optimal policy choice. 
My paper complements this tool-set by giving the policy-maker an explicit measure of robustness of a policy claim to shifts in the covariate distributions. There are two reasons to focus on observable characteristics. First, they are readily available to the policy-maker and are likely to be of first-order importance when assessing the robustness of (quasi)-experimental findings. Second, the resulting robustness metric is identified through the (quasi)-experimental data, limiting the need for bounding or partial identification approaches.  \par

The paper is organized as follows: Section \ref{section:Heterogeneity and Robustness} introduces the setup and defines the robustness metric. Section \ref{section:Estimation and Asymptotic Results} presents the estimator and its asymptotic properties. Section \ref{section:Empirical Application} applies the approach to the findings from the Oregon health insurance experiment and provides interpretations for the robustness measure. Section \ref{section:Conclusion} concludes. The main proofs are in the Appendix. Additional results and proofs of the lemmas are available in \citet{spini_robustness}. 
Throughout, all results apply to both experimental and quasi-experimental settings. I use ``experiment'' to refer to both interchangeably.%
\section{A robustness metric for covariate shifts}
\label{section:Heterogeneity and Robustness}
In this section, I link treatment-effect heterogeneity to robustness in a potential-outcomes framework, focusing on the Average Treatment Effects (ATE). A policy-maker is interested in whether a claim such as $ATE > \tilde{\tau}$ remains valid when the covariate distribution differs from the experiment. Using the Conditional Average Treatment Effect (CATE), which aggregates to the ATE, I characterize the covariate distribution that (i) violates the claim and (ii) is closest to the experimental distribution. I call it the \emph{least-favorable distribution} because it is the hardest to distinguish from the experimental distribution. Distance is measured by the KL divergence: the resulting minimum distance, $\delta^*$, is my robustness metric. Any covariate distribution within KL distance $\delta^*$ of the experimental distribution necessarily preserves the claim.

\subsection{Set-up and Preliminaries}

Let $Y_d$ denote the potential outcome under binary treatment $d$. In the experiment, for each unit, one observes treatment status $D\in\{0,1\}$, realized outcome $Y=DY_1+(1-D)Y_0 \in\mathcal Y$, and covariates $X\in\mathcal X$.\footnote{Additional control variables $W$ can be accommodated. I suppress them here to simplify notation.} I partition the covariates as $X=(X_c,X_e)$, where $X_c$ collects covariates with counterparts in census-type data in other states, and $X_e$ collects experiment-specific covariates with no natural counterpart outside the experiment. In my empirical application, $X_c$ includes indicators for race, gender, age, education, and urban area whereas $X_e$ includes \citet{finkelstein2012oregon}'s proxy for health status.
Though treatment effect heterogeneity will be estimated using the full vector $X$, its partition into $(X_c,X_e)$ will play a role in my benchmarking exercise to gauge the magnitude of my robustness metric. Let $P_X$ denote the probability measure for $X$ in the experiment and let $F_X$ be its associated distribution. I distinguish potential outcomes under $F_X$ from those under an alternative $F'_X$ by writing $Y_d$ and $Y'_d$, respectively. The propensity score is $\pi(x)=P(D=1\mid X=x)$. Finally, for any random variable $W$, let $\mathcal W$ denote its support. The interior of a set $S$ is $\interior{S}$.

\begin{assumption}\label{ass:Unconfoundedness}{Unconfoundedness and Overlap}
\begin{enumerate}[label=\roman*)]
\item $Y_1,Y_0 \indep D | X$. 
\item There exists an $\epsilon >0$ such that all $x \in \mathcal{X}$ we have $0 < \epsilon \leq \pi(x) \leq 1 - \epsilon < 1$
\end{enumerate} 
\end{assumption}
In an RCT, under complete or covariate-based randomization of treatment assignment, Assumption \ref{ass:Unconfoundedness} i) holds by design. In the case of quasi-experimental studies Assumption \ref{ass:Unconfoundedness} i) requires the researcher to carefully evaluate the selection mechanism that governs program participation. 
Assumption \ref{ass:Unconfoundedness} ii) is strict overlap: weaker forms still allow identification, this version is needed for estimation in Section \ref{section:Estimation and Asymptotic Results}. 

The policy-maker's parameter of interest is the $ATE := \mathbb{E}[Y_1 - Y_0]$. The CATE, defined by $\tau(x) := CATE(x) = \mathbb{E}[Y_1 - Y_0|X=x]$, captures how average effects change across sub-populations with covariate value $X =x$. Under Assumption \ref{ass:Unconfoundedness} i), $\tau_{F_X}(x)$ is nonparametrically identified by the difference between $\gamma_1(X) := \mathbb{E}[Y|D=1,X=x]$ and $\gamma_0(X) =  \mathbb{E}[Y|D=0,X=x]$ using data from the experiment \citep{imbens2015causal}.\footnote{If the CATE only partially identified, like in the case of non-compliance based on unobservables, it is possible to follow a bounding approach for my robustness procedure. This approach is sketched in \citet{spini_robustness} but I leave the details for future research.} ATE is obtained by averaging $\tau(x)$ with weights proportional to $F_X$. We can write the map sending $F_X$ to its corresponding ATE as: 
\begin{equation}
\label{ATE mapping}
ATE : F_{X} \mapsto \int_{\mathcal{X}} \tau_{F_X}(x) dF_{X}(x)
\end{equation}
The subscript $F_X$ on $\tau(x)$ indicates that, in general, it's possible that the functional form of CATE depends on $F_X$. In this case, a change in the distribution of the covariates from $F_X$ to $F'_X$ would affect the magnitude of ATE through two channels:  a direct effect through the weights of $F_X(x)$ and an indirect effect through changing the shape of the function $x \mapsto \tau(x)$. Of course, without further assumptions, $\tau_{F_X}(x)$ is only identified when $F_X$ is the experimental distribution. In this paper, I use the covariate shift assumption\footnote{This assumption appears, for example, also in \citet{hsu2020counterfactual} and \citet{jeong2020robust}.} to eliminate the indirect effect.

\begin{assumption}\label{ass:Covariate Shift}{\textbf{(Covariate Shift})}
Let $X'$ denote the covariates in the new environment. Then:
\begin{enumerate}[label=\roman*)]
\item  $F_{Y_d'|X'}(y|x) = F_{Y_d|X}(y|x)$ for $d \in \{ 0,1 \} $, for all $x \in \mathcal{X}$ and $y \in \mathcal{Y}_d$ and all distributions of $X'$. 
\item  $\mathcal{X}' \subseteq \mathcal{X}$
\end{enumerate}
\end{assumption}
Assumption \ref{ass:Covariate Shift} i) says that the causal link between the treatment variable $D$ and the potential outcomes of interest $Y_1$ and $Y_0$ does not depend on the distribution of the observables. One could think of Assumption \ref{ass:Covariate Shift} i) as analogous to a policy invariance condition with respect to the distribution of covariates. 
Assumption \ref{ass:Covariate Shift} ii) says the support of the covariates in the new environments is contained in the support of the baseline environment. In practice, this limits the extrapolation to environments for which any value of the covariates could have been observed in the experimental setting as well, albeit with a different weight. 
Because Assumption \ref{ass:Covariate Shift} guarantees that $\tau_{F_X}(x)$, the CATE, does not vary when $F_X$ is replaced by any other distribution $F_{X'}$ it is not necessary to index $\tau(x)$ with $F_X$.\footnote{This identification result follows immediately from Ass. \ref{ass:Covariate Shift}. See \citet{hsu2020counterfactual}, Lemma 2.1.} Then, the link between $F_X$ and ATE reduces to integration against a fixed $\tau(x)$ and the map in Eq.\eqref{ATE mapping} is linear in $F_X$:
\begin{equation}
\label{ATE mapping noindex}
ATE : F_{X} \mapsto \int_{\mathcal{X}} \tau(x) dF_{X}(x)
\end{equation}
Before presenting the general framework I give the simplest nontrivial example of a robustness exercise with respect to covariate shifts.

\begin{example}
\label{BinaryExample}
Consider a binary covariate $X = \{ 0,1 \}$. $D$ is randomized,
satisfying Assumption \ref{ass:Unconfoundedness}. Then $\tau(x):=\mathbb{E}[Y_1|X=x]-\mathbb{E}[Y_0|X=x]$ is identified. Because $X$ is Bernoulli, any distribution on $\{0,1\}$ is fully described by $P_X(X=1) = p_1$. 
\begin{align*}
ATE(F_X) = ATE(p_1)
&= \tau(0) \cdot (1-p_1) + \tau(1) \cdot p_1.
\end{align*}
Suppose the experimental $ATE > 0$ and $\tau(1) > 0 > \tau(0)$: treatment effects have different signs. What is the closest covariate distribution invalidating the claim $ATE > 0$? Setting ATE to 0 and solving for $p_1^*$: 
\begin{equation*}
\tau(0) \cdot (1-p_1^*) + \tau(1) \cdot p_1^* = 0 \implies
p_1^* = \frac{-\tau(0)}{\tau(1) - \tau(0)} \in [0,1].
\end{equation*}
$|p_1^*- p_1| = |\frac{-\tau(0)}{\tau(1) - \tau(0)} - p_1|$ is the smallest shift in $p_1$ that invalidates the claim $ATE > 0$.
\end{example}
Under what conditions does a solution like $p_1^*$ exist in general, and when is it unique? When $\mathcal{X}$ is not discrete, distributions on $\mathcal{X}$ are inherently infinite-dimensional unless one imposes parametric restrictions. Moreover, how should one measure the distance between $p_1^*$ and $p_1$ in this case? Motivated by these questions, I adopt a nonparametric notion of distance between probability distributions.\footnote{\citet{spini_robustness} discusses how the general procedure can be specialized to certain parametric classes of distributions. In such cases, the relevant covariate shifts coincide with mean shifts.}
 
\begin{definition}[KL-divergence]
\label{def:KL-divergences}
The $KL$-divergence between two distributions $F_X$ and $F_X'$ is given by:
\begin{equation}
D_{KL}(F_X'||F_X) :=  \int_{\mathcal{X}} \log \left(\frac{dF_X'}{dF_X}(x) \right) \frac{dF_X'}{dF_X}(x) dF_X(x),
\end{equation}
where $\frac{dF_X'}{dF_X}$ is the Radon-Nikodym derivative of distribution $F_X'$ w.r.t the experimental distribution $F_X$, provided that $P_X' \ll P_X$ for the respective probability measures. \end{definition}

\textbf{Why choose KL?} There are several ways to measure the distance between a candidate distribution $F'_X$ and $F_X$, each with advantages and disadvantages. For our purposes, a first requirement is to measure discrepancies without committing to a parametric family $\mathcal F$. As noted in the introduction, unless the shape of $\tau(x)$ is restricted jointly with $\mathcal F$, one can change the ATE while keeping a parametric discrepancy arbitrarily small. This requirement rules out, for example, Mahalanobis-type measures and motivates focusing on nonparametric distances. 

Two prominent classes of nonparametric distances are Integral Probability Metrics (IPMs) and $\phi$-divergences. IPMs include Total Variation, kernel-based distances like MMD and, notably, Wasserstein distances (e.g., $W_1$, $W_2$) which have been increasingly adopted in recent work (e.g., \citet{adjaho2022externally,athey2024using,gunsilius2023distributional}). For our problem of quantifying robustness to covariate shifts, Wasserstein is less attractive for two reasons. First, Wasserstein is generally not invariant to invertible transformations of the covariates $X$. More broadly, under standard regularity conditions, requiring an IPM to satisfy this invariance property is highly restrictive and, in effect, forces the distance to be a $\phi$-divergence \citep{qiao2010study}. As a consequence of this lack of invariance, transformations commonly used in empirical work, such as logs or inverse hyperbolic sine, (which are not isometries), can mechanically change the measured Wasserstein distance and the resulting robustness assessment. Second, Wasserstein-based optimization typically relies on entropic regularization to inject sufficient convexity for computational tractability.\footnote{For results on convergence rates for entropy-regularized Wasserstein distance see \citet{tabri2026sieve}.} In contrast, the KL formulation naturally enjoys this convexity and, in addition, delivers a sharp closed-form characterization as discussed in Section~\ref{subsection:A closed form solution for quantifying robustness}. 

The second prominent class is the family of $\phi$-divergences (including $\chi^2$, Hellinger, and KL). Unlike Wasserstein, all $\phi$-divergences are invariant to any invertible measurable transformation of $X$ (with measurable inverse).\footnote{One can show the following result. Let $T: \mathcal{X} \rightarrow T(\mathcal{X})$ be a (deterministic) measurable transformation with measurable inverse $T^{-1}(\cdot)$. Then $D_{KL}(G_{T(X)} \Vert F_{T(X)})=D_{KL}(G_{X} \Vert F_{X})$.} In addition, KL has three advantages that make it a particularly attractive choice for our problem. First, under KL the minimization problem in Eq.\eqref{eq:Inf_Problem}-\eqref{eq:Constraint} has an exponentially tilted, unique interior solution (Theorem~\ref{thm:least favorable features}), which in turn enables convenient (de-biased) GMM estimation and interpretation of the minimizer (discussed in Theorem \ref{thm:least favorable features}). Second, KL satisfies an exact chain-rule decomposition over conditional distributions, which I exploit in the benchmarking exercise in Section \ref{subsection: Benchmarking Robustness}. Third, KL’s natural connection to hypothesis testing and power lets us quantify ``distance" in terms of statistical distinguishability: a KL magnitude can be mapped to the minimal sample size needed to detect a covariate shift. I illustrate this connection in Section~\ref{subsec:Benchmarking empirical}. 

\subsection{The policy-maker's problem: quantifying robustness}
\label{subsection:the policy-maker's problem: quantifying robustness}
Once the $ATE$–covariate link and a distance measure are specified, I can formalize the policy-maker's robustness problem. Consider the claim given by $ATE > \tilde{\tau}$ (the reverse inequality is analogous) where $\tilde{\tau}$ is a policy-relevant threshold capturing, for example, implementation costs or the value of a competing policy. In Example \ref{BinaryExample} $\tilde{\tau}=0$, often a natural benchmark. The policy-maker is interested in the smallest shift from the experimental distribution, $F_X$, invalidating the claim $ATE > \tilde{\tau}$. Formally:  
\begin{align}
\inf_{F_X': \ P_X' \ll P_X; P_X'(\mathcal{X}) = 1} & D_{KL}(F_X'||F_X) \label{eq:Inf_Problem} \\
s.t. \ & \int_{\mathcal{X}} \tau(x) d F_X'(x) \leq \tilde{\tau}. \label{eq:Constraint}
\end{align}
Eq.\eqref{eq:Inf_Problem}–\eqref{eq:Constraint} define an optimization problem over the set of covariate distributions violating the claim $ATE > \tilde{\tau}$ and the objective selects the one closest to $F_X$ in KL divergence. By Assumption~\ref{ass:Covariate Shift}, $\tau(x)$ in Eq.\eqref{eq:Constraint} is not indexed by $F_X'$: the constraint set is linear in $F_X'$. Since $D_{KL}(\cdot \Vert F_X)$ is a strictly convex function on a convex feasible set, the problem admits a unique ($P_X$–a.e.) solution, characterized in Theorem~\ref{thm:closed form solution}.

\begin{remark}
The program in Eq.\eqref{eq:Inf_Problem}--\eqref{eq:Constraint} is nonparametric: no structure beyond absolute continuity $P_X'\ll P_X$ is imposed on $F_X'$.\footnote{For example, if both $P_X$ and $P_X'$ are absolutely continuous with respect to a common dominating measure $\lambda$ and $\mathrm{supp}(P_X')\subseteq \mathrm{supp}(P_X)$, then $P_X'\ll P_X$.} 
Absolute continuity implies that $F_X'$ cannot assign mass outside the experimental support $\mathcal X$, which I view as natural: the feasible covariate distributions should not place weight on subpopulations for which the experiment provides no data. Similar support restrictions are standard in the distributional policy effects literature, e.g.,\ \citet{rothe2012partial}. Without them, $\tau(x)$ is unidentified at covariate values $x$ that are never observed and can be chosen to generate arbitrarily large average effects, rendering the robustness exercise uninformative.
\end{remark}

\noindent We are now ready to define the \textit{least-favorable distribution} and the robustness metric.

\begin{definition}
\label{def:LeastFavorable&Robustness}
Fix $\tilde{\tau}$. The robustness metric $\delta^*(\tilde{\tau})$ is the infimum of Eq.\eqref{eq:Inf_Problem}. The (set of) \textit{least-favorable distribution(s)} $\{F_X^*(\tilde{\tau})\}$ is the (set of) minimizer(s) of Eq.\eqref{eq:Inf_Problem}. 
\end{definition}

I define $\delta^*(\tilde{\tau})$ as the KL distance between the experimental covariate distribution $F_X$ and a \emph{least-favorable} distribution $F_X^*$, i.e., the KL-closest covariate shift that violates the claim $ATE>\tilde{\tau}$. As I show in Section \ref{subsec:Benchmarking empirical}, among all shifts that invalidate the claim, this least-favorable shift is the hardest to detect statistically. The set $\{F_X^*(\tilde{\tau})\}$ may be empty for some values of $\tilde{\tau}$. When the experimental ATE already violates the target (i.e., $ATE(F_X)\le \tilde{\tau}$), the experimental distribution is feasible and achieves the minimum: one can take $F_X^*=F_X$, so $\delta^*(\tilde{\tau})=0$. The problem is nontrivial when $ATE(F_X)>\tilde{\tau}$, in which case $F_X$ is infeasible for \eqref{eq:Constraint} and, under regularity conditions, any least-favorable solution must satisfy $\delta^*(\tilde{\tau})>0$. In Example~\ref{BinaryExample}, I imposed $ATE(p_1)>0$ precisely to ensure this nontrivial case.

With discrete $\mathcal X$, $F_X$ is a point in the probability simplex. The optimization
in Eq.\eqref{eq:Inf_Problem}-\eqref{eq:Constraint} has a simple geometry: the constraint $\int_{\mathcal X}\tau(x)\,dF_X'(x)\le \tilde{\tau}$
defines the set of covariate shifts invalidating the claim, and the KL objective selects the closest one.

\begin{example}\label{ex:threepoints}
Let $\mathcal X=\{H,M,L\}$ index income bins and write $F_X'=(p_H,p_M,1-p_H-p_M)$. Suppose $(\tau(H),\tau(M),\tau(L))=(1,2,3)$ and set the threshold $\tilde{\tau}=1.8$.
Under the experimental distribution, $ATE(F_X)=2.4>\tilde{\tau}$.
\end{example}

Figure~\ref{fig:Three discrete variables} depicts KL level sets around $F_X$, the feasible region, and the \textit{least-favorable distribution} from Example~\ref{ex:threepoints}. Since the objective and constraint in Eq.\eqref{eq:Inf_Problem}-\eqref{eq:Constraint} are smooth in $(p_H,p_M)$, the solution is characterized by the KKT conditions. The green curve is the KL level set at $\delta^*(\tilde{\tau})$: any covariate distribution closer than $\delta^*(\tilde{\tau})$ is guaranteed to satisfy the policy-maker's claim. When $\mathcal X$ is not discrete, a visualization like Figure~\ref{fig:Three discrete variables} may be difficult. Nonetheless, under mild regularity
conditions, a \textit{least-favorable distribution} exists, is unique and admits a closed-form characterization via exponential tilting, with a well-defined
$\delta^*(\tilde{\tau})$ for a range of values of $\tilde{\tau}$. 

\begin{figure}[!t]
    \centerline{
    \adjincludegraphics[width= 0.75\textwidth,Clip={2cm} {0.5cm} {2cm} {0.8cm}]{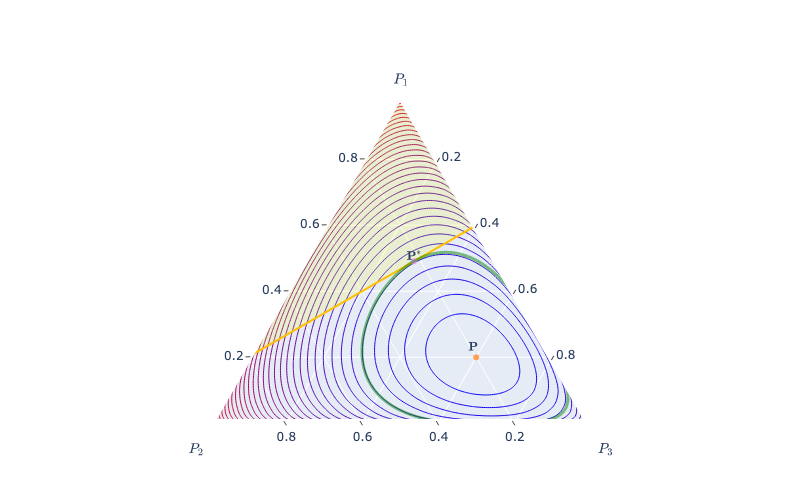}}
    \caption[\footnotesize Discrete-support example: geometric intuition for the least-favorable distribution]{\footnotesize
The triangle is the 2-simplex of distributions $P'=(p_H,p_M,p_L)$ on $\mathcal X=\{H,M,L\}$ (barycentric coordinates). The experimental distribution is $P=(0.2,0.2,0.6)$. The shaded region consists of distributions with $ATE(P')\le 1.8$. Contours are KL level sets $D_{KL}(P'\Vert P)$ (darker = smaller). The point $P^*=(0.491,0.218,0.291)$ is the \textit{least-favorable distribution}. The green contour is the level set at $\delta^*\approx 0.296$; any $P'$ with $D_{KL}(P'\Vert P)<\delta^*$ is guaranteed to satisfy the claim.}
\label{fig:Three discrete variables}
\end{figure}

\subsection{A closed form solution for quantifying robustness}
\label{subsection:A closed form solution for quantifying robustness}

I characterize the solution to Eq.\eqref{eq:Inf_Problem}-\ref{eq:Constraint} under regularity conditions.

\begin{assumption}[Boundedness]\label{ass:boundedness}
The conditional average treatment effect $\tau(X)$ is bounded $P_X$-a.s. (i.e., there exists $M<\infty$ such that $\mathbb{P}_{X} (|\tau(X)| \leq M) =1$).
\end{assumption}

Assumption~\ref{ass:boundedness} also holds $P'_X$-a.s.\ for any $P'_X\ll P_X$, since absolute continuity prevents $P'_X$ from assigning mass to $P_X$-null sets. Boundedness as in Assumption \ref{ass:boundedness} is plausible in a cross sectional, micro-econometrics setting. One could weaken it to an exponential-moment condition, e.g.,\ $\mathbb E[\exp(\kappa|\tau(X)|)]<\infty$ for all $\kappa>0$ (see \citealp{komunjer2016existence}). I maintain Assumption~\ref{ass:boundedness} for ease of exposition. 

The constraint set $\{F'_X:\int \tau\,dF'_X\le \tilde{\tau}\}$ is a convex set. If empty, it implies a value of $+\infty$ for the program in Eq.\eqref{eq:Inf_Problem}-\eqref{eq:Constraint}. To rule this out, I assume $\tilde{\tau}$ is attainable via some covariate shift, i.e.,\ there exists an $F'_X$ with $ATE(F'_X)=\tilde{\tau}$. If $\tau(x)$ is not sufficiently heterogeneous, like in Example \ref{ex:homogeneous treatment effects}, ATE can never reach $\tilde{\tau}$.

\begin{example}[Homogeneous treatment effects]\label{ex:homogeneous treatment effects}
If $\tau(x)= c$, for some $c \in \mathbb{R}$, then $ATE(F'_X)=c$ for every $F'_X$: no covariate shift can reach a different threshold $\tilde{\tau}\neq c$.
\end{example}

Constant treatment effects are an extreme case: under Assumption~\ref{ass:Covariate Shift}, if $\tau(x)= c$ then $ATE(F'_X)=c$ for every $F'_X$, so the claim $ATE > \tilde{\tau}$ can be extrapolated from the experiment to any environment. More generally, if $\tilde{\tau}$ lies outside the range of $\tau(X)$, then no covariate distribution can satisfy $ATE=\tilde{\tau}$, the feasible set in Eq.\eqref{eq:Constraint} is empty (hence the value of Eq.\eqref{eq:Inf_Problem} is $+\infty$). For instance, if $2\le \tau(X)\le 5$ a.s., then $\tilde{\tau}=1$ is unattainable. Non-trivial robustness requires some minimal heterogeneity of treatment effects.\footnote{It is convenient to restrict the robustness metric to a real valued parameter rather than $\mathbb{R} \cup \{+\infty\}$.} The following assumption guarantees a nonempty feasible set:

\begin{assumption}(Non-emptiness)
\label{ass:Non-emptiness}
Recall $\mathcal F := \{ F_X' \ s.t. \ P_X' \ll P_X \}$.
Define \(L(F_X) := \int_{\mathcal X} \tau(x)\, dF_X(x)\).
I require $\tilde{\tau} \in \interior L(\mathcal F)$.
\end{assumption}
Assumption \ref{ass:Non-emptiness} is a joint condition on $\tau(\cdot)$ and $\tilde{\tau}$. It requires $\tau(\cdot)$ to be sufficiently heterogeneous so that some $F'_X\in\mathcal F$ attains $ATE(F'_X)=\tilde{\tau}$. In Example \ref{BinaryExample}, Assumption \ref{ass:Non-emptiness} is satisfied: $\tau(0) < 0 < \tau(1)$ guaranteeing $\tilde{\tau} = 0 \in \interior L(\mathcal F) = (\tau(0),\tau(1))$. In Example \ref{ex:homogeneous treatment effects}, it fails: $L(\mathcal F) = \{c\}$ so $\interior L(\mathcal F) = \varnothing$. More broadly, the set $L(\mathcal F)$ captures how rich is the set of ATEs that could be attained by a covariate shift $F'_X$. \\
Assumption \ref{ass:Non-emptiness} also has testable implications. Given an estimate of $\tau(x)$, one can test whether $\tilde{\tau}\in(\inf_x \tau(x),\sup_x \tau(x))$ using the intersection-bounds procedure in \citet{chernozhukov2013intersection}. Homogeneous (or nearly homogeneous) treatment effects are cases in which such tests may reject. Technically, Assumption \ref{ass:Non-emptiness} serves as a (Slater-like) constraint qualification guaranteeing strong duality and the existence of the \textit{least-favorable distribution} (see \citet{komunjer2016existence} for an excellent overview on optimization problems involving KL and other $\varphi$-divergences).\footnote{The interior condition cannot be weakened. Under Ass. \ref{ass:boundedness}, $L(\mathcal{F})$ is an interval. If $\tilde{\tau}$ lies on the boundary, the feasible set in Eq.\eqref{eq:Constraint} collapses to a degenerate measure w.r.t $F_X$, implying infinite KL in Eq.\eqref{eq:Inf_Problem}. See the quasi-relative interior condition in \citet{borwein1993partially} Eq.(BL).} After introducing Assumptions \ref{ass:Unconfoundedness}-\ref{ass:Non-emptiness} we can state the key result below: 
\begin{theorem}[Closed form solution]
\label{thm:closed form solution}
Let Assumptions \ref{ass:Unconfoundedness}-\ref{ass:Non-emptiness} hold. Then:
 i) Eq.\eqref{eq:Inf_Problem} attains a minimum at $F_X^*$, uniquely characterized $P_X$-almost everywhere by:
\begin{equation}
\label{eq:Worst Case Distribution} 
\frac{dF_X^*}{d F_X}(x) = \frac{\exp(-\lambda (\tau(x) - \tilde{\tau}))}{\int_{\mathcal{X}} \exp(-\lambda (\tau(x) - \tilde{\tau})) dF_X(x)},     
\end{equation}
where $\lambda$ is the unique Lagrange multiplier implicitly defined by the equation:
\begin{equation}
\label{eq:Lagrange Multiplier}
\int_{\mathcal{X}} \exp(-\lambda (\tau(x) - \tilde{\tau}))(\tau(x) - \tilde{\tau}) dF_X(x) = 0.
\end{equation}
ii) The robustness metric $\delta^*(\tilde{\tau})$ is finite and given by:
\begin{equation}
\label{eq:ClosedFormDelta}
\delta^*(\tilde{\tau}) = D_{KL}(F_X^*||F_X) = - \log\left(\int_\mathcal{X} \exp(-\lambda(\tau(x)- \tilde{\tau})) dF_X(x) \right).
\end{equation}
\end{theorem}
Theorem \ref{thm:closed form solution} shows that the general robustness problem admits an exponential-tilting solution, making optimization over the nonparametric class $\mathcal{F}$ no harder than the parametric Examples \ref{BinaryExample} and \ref{ex:threepoints}. Their solution from Theorem \ref{thm:closed form solution} coincides with the one obtained directly by the KKT conditions. The theorem also implies $\delta^*(\tilde{\tau})$ is identified as a functional of $(F_X,\tau(\cdot))$, both nonparametrically identified from the experimental data, and it motivates the estimation approach in Section~\ref{section:Estimation and Asymptotic Results}. 

$\delta^*(\tilde{\tau})$ is distinct from other popular statistics used to summarize treatment-effect heterogeneity. Relative to rich descriptive summaries like sorted effects \citep{sorted_chernozhukov2018} and conditional quantile treatment effects \citep{koenker2005quantile}, $\delta^*(\tilde{\tau})$ is more parsimonious, collapsing heterogeneity into a single, easy-to-report number. Relative to other scalar summaries, $\delta^*(\tilde{\tau})$ is directly geared to the external-validity exercise: $\textrm{Var}(\tau(X))$ is threshold-free and therefore not tied to a particular claim, though it can be locally informative about robustness when $ATE_{F_X}$ lies near $\tilde{\tau}$. Likewise, the tail mass $P(\tau(X)\le \tilde{\tau})$ is threshold-specific but coarse, as it ignores changes in the shape of $\tau(X)$ that do not move mass across $\tilde{\tau}$. Overall, these trade-offs reflect complementary goals for these summaries. In this sense, $\delta^*(\tilde{\tau})$ expands the researcher's toolkit with a simple way to gauge external validity under covariate shifts.

\subsection{Benchmarking and interpreting robustness}
\label{subsection: Benchmarking Robustness}
After a researcher reports $\delta^*(\tilde{\tau})$, how should they interpret its magnitude? I propose to benchmark $\delta^*(\tilde{\tau})$ against the covariate shifts one could observe in other implementation environments, when such data is available. I leverage a decomposition of KL for the partition  $X=(X_c,X_e)$ leading to the useful result \citep{cover1999elements}:
\begin{proposition}
\label{prop:KL_decomposition}
Let $X=(X_c,X_e)$. Then for any two distributions $F$ and $G$: 
\begin{equation}
\label{eq:KL_decomposition}
D_{KL}(G_X \Vert F_X) = D_{KL}(G_{X_c} \lVert F_{X_c}) + \mathbb{E}_{G_{X_c}} \left( D_{KL}(G_{X_e|X_c} \lVert F_{X_e|X_c}) \right).   
\end{equation}
If $X_c \indep X_e$ for both $F$ and $G$, then $D_{KL}(G_X \Vert F_X) = D_{KL}(G_{X_c} \lVert F_{X_c}) + D_{KL}(G_{X_e} \lVert F_{X_e})$.
\end{proposition}
Proposition \ref{prop:KL_decomposition} decomposes  the left-hand side into two additive terms: (1) the KL between marginals $F_{X_c}$ and $G_{X_c}$ and (2) the average KL between the conditional distributions $F_{X_e|X_c}$ and $G_{X_e|X_c}$. Because the second term in Eq.\eqref{eq:KL_decomposition} is the expectation of a non-negative quantity, we have the lower bound
$D_{KL}(G_{X_c}\Vert F_{X_c}) \le D_{KL}(G_X\Vert F_X)$ (with equality if and only if
$G_{X_e\mid X_c}=F_{X_e\mid X_c}$ $G_{X_c}$-a.s.). Eq.\eqref{eq:KL_decomposition} yields a simple benchmarking strategy for $\delta^*(\tilde{\tau})=D_{KL}(F_X^*\Vert F_X)$. For any candidate implementation environment where $X_c$ is observed, we can compute the first term. Because the second term is not identified when $X_e$ is unavailable outside the experiment, a practical approach is to calibrate its magnitude as a proportion of the first term so that $D_{KL}(G_X\Vert F_X)=(1+\kappa)\,D_{KL}(G_{X_c}\Vert F_{X_c})$. Since both are in KL units, the comparison supports cardinal statements such as \textit{``$\delta^*(\tilde{\tau})$ is 20\% larger than the KL distance based on the census covariates $X_c$ between Texas and the Oregon experiment"}. If this $\kappa$ calibration is valid, any environment with $(1+\kappa)\,D_{KL}(G_{X_c}\Vert F_{X_c})<\delta^*(\tilde{\tau})$ is guaranteed to preserve the experimental claim. The special case $\kappa=0$ corresponds to the assumption $G_{X_e\mid X_c}=F_{X_e\mid X_c}$: the conditional distribution of the missing covariates matches the experiment.
Section~\ref{subsec:Benchmarking empirical} illustrates this benchmarking exercise for the empirical application in \citet{finkelstein2012oregon}, where other states are the new implementation environments, $X_c$ includes race, gender, age, education, and urban indicators (available in the state census), while $X_e$ is a proxy for health status (available only in the experiment). In the empirical application, I consider $\kappa\in\{0,0.2,1\}$, corresponding to no unobserved contribution, a contribution proportional to the relative dimension of $X_e$ and $X_c$, and equal observed and unobserved contributions.
\section{Estimation and Asymptotic Results} \label{section:Estimation and Asymptotic Results}
This section develops semiparametric estimators and establishes the asymptotic properties of: (i) the robustness metric 
$\delta^*$ in Equation~\eqref{eq:ClosedFormDelta} and (ii) of user-specified moments of the least favorable distribution in Equation~\eqref{eq:Worst Case Distribution}. 
The analysis builds on the de-biased GMM framework of \citet{chernozhukov2020locally}. Let $W = (Y,D,X)$ be the data. The characterization of Eq.\eqref{eq:Lagrange Multiplier}-\eqref{eq:ClosedFormDelta}- in Theorem \ref{thm:closed form solution} as integrals w.r.t $F_X$ implies a moment condition. As in \citet{newey1994chapter}, define a parameter $\theta_0 := (\nu_0,\lambda_0)^T$ satisfying the population moment: 
\begin{equation}
\label{eq:Moment conditions}
\mathbb{E}[g(W,\theta,\tau)] = \mathbb{E} \begin{bmatrix} 
\exp(-\lambda_0(\tau_0(X) - \tilde{\tau})) - \nu_0  \\
\exp( -\lambda_0( \tau_0(X) - \tilde{\tau})) (\tau_0(X) - \tilde{\tau}) \end{bmatrix} = \begin{bmatrix} 0 \\
0 \end{bmatrix},
\end{equation}
where $\tau_0(\cdot)$ is the true CATE. Assumptions \ref{ass:Unconfoundedness}--\ref{ass:Non-emptiness} guarantee that $(\nu_0,\lambda_0)$ are globally identified by Equation \eqref{eq:Moment conditions}.
From these, the robustness metric in Equation \eqref{eq:ClosedFormDelta} is identified by $\delta^* = -\log(\nu_0)$. The parameter space $\Theta \subseteq \mathbb{R}^2$ satisfies $0<\nu_0\le 1$. Positivity follows by $\exp(-\lambda_0(\tau(x)-\tilde{\tau}))>0$ for all $x$. If $ATE > \tilde{\tau}$ holds under $F_X$, then \(\delta^*>0\) and hence $\nu_0<1$. Because the true $\tau_0(X)$ is unknown but identified, a feasible version of Equation \eqref{eq:Moment conditions} replaces  $\tau_0(X)$ with a first step estimate $\hat{\tau}(X)$. The plug-in GMM $\hat{\theta} = (\hat{\lambda}, \hat{\nu})^T$ solves the sample equivalent of Equation \ref{eq:Moment conditions}.

Assumption \ref{ass:Unconfoundedness} guarantees nonparametric identification of $\tau_0(X)$ via $\gamma_1(X)-\gamma_0(X)$, leaving the researcher free to choose among many estimation strategies, both parametric and nonparametric. Popular options include random forest \citep{athey2016approximate} or doubly robust scores \citep{hsu2020counterfactual}. Because moment conditions in Equation \eqref{eq:Moment conditions} are not Neyman orthogonal w.r.t $\hat{\tau}(X)$, the first-step estimation of $\hat{\tau}(X)$ may affect the asymptotic distribution of $\hat{\delta}^* = -\log(\hat{\nu})$ through the plug-in estimator. This can lead to invalid inference unless one imposes some \textit{ad-hoc} rate conditions on the first step, which are restrictive and difficult to verify in practice \citep{chernozhukov2018double}. To avoid that, following \citet{chernozhukov2020locally}, I derive a debiased-GMM estimator that yields orthogonal moments and valid inference under a range of flexible estimators for $\tau_0(X)$. Details of this derivation are in \citet{spini_robustness}. 
It is convenient to index any functional nuisance parameter like $\gamma_1(X)$ and $\gamma_0(X)$ with the distribution of the data, with $F_0$ denoting the true distribution.

\begin{proposition}
\label{prop:De-biased GMM nonparamteric influence function}
The nonparametric influence function based on $g(\cdot)$ in Eq.\eqref{eq:Moment conditions} is:
\begin{align}
\label{eq:Non-paramteric Influence function characterization}
\phi(w,\theta,\gamma_0,\alpha_0) &= \begin{bmatrix} \exp \left( -\lambda \cdot (\gamma_{1,F_0}(x) - \gamma_{0,F_0}(x) - \tilde{\tau} ) \right) \cdot (-\lambda) \\
 \exp \left( -\lambda \cdot (\gamma_{1,F_0}(x) - \gamma_{0,F_0}(x) - \tilde{\tau} ) \right) \cdot (1-\lambda \cdot (\gamma_{1,F_0}(x) - \gamma_{0,F_0}(x) - \tilde{\tau} )) \end{bmatrix} \nonumber \\
&\quad\times
\Bigl(
  \alpha_{1,F_0}(x)\, d\,(y - \gamma_{1,F_0}(x))
  + \alpha_{0,F_0}(x)\,(1-d)\,(y - \gamma_{0,F_0}(x))
\Bigr) \nonumber \\
&\alpha_{F_0} :=(\alpha_{1,F_0}(x),  \alpha_{0,F_0}(x)) = \left(\frac{1}{\pi_{F_0}(x)}, \frac{1}{1-\pi_{F_0}(x)}\right),
\end{align}
\end{proposition}

\begin{proposition}
\label{prop:Neyman orthogonality}
The de-biased moment functions below are Neyman orthogonal. 
\begin{equation}
\label{eq:De-biased GMM}
\psi(w,\gamma,\theta, \alpha) = g(w,\theta,\gamma) + \phi(w,\theta,\gamma,\alpha).  
\end{equation}
\end{proposition}
Note that $\mathbb{E}_{F_0}[\psi(W,\theta,\gamma_0,\alpha_0)] = 0$. Under regularity conditions $\mathbb{V}[\psi(W,\theta,\gamma_0,\alpha_0)] < \infty$,  $\psi(\cdot)$ is a valid influence function. The $K$-fold cross-fitted de-biased GMM equations from Eq.\eqref{eq:De-biased GMM}, and the corresponding estimator for $\theta$ are:
\begin{align}
\hat{\psi}(\theta,\hat{\gamma},\hat{\alpha}) &= \frac{1}{K} \sum_{k = 1}^K \frac{1}{|I_k|} \sum_{i \in I_k} \left( g(W_i, \theta, \hat{\gamma}_{-k}) +  \phi(W_i, \tilde{\theta}, \hat{\gamma}_{-k}, \hat{\alpha}_{-k}) \right) \\
\label{eq:debiasedGMMtheta}
\hat{\theta} &= \arg \min_{\theta \in \Theta} \hat{\psi}(\theta, \hat{\gamma},\hat{\alpha}),
\end{align}
where $\tilde{\theta}$ is a consistent estimator for $\theta$ needed to evaluate $\phi(\cdot)$.\footnote{A natural candidate is the plug-in GMM, which is $o_p(1)$ but may not be $o_p(n^{-\frac{1}{2}})$ in general.} For each fold $k=1,\cdots, K$, $\gamma(\cdot)$ and $\alpha(\cdot)$ are estimated on the remaining $(K-1)$ folds, then the empirical moment is evaluated on fold $k$. Sample-splitting reduces own-observation bias and, together with the Neyman orthogonality, avoids complicated Donsker-type conditions that could fail for some estimators of $\hat{\gamma}$ and $\hat{\alpha}$, as discussed in \citet{chernozhukov2020locally}. 
To establish $\sqrt{n}$-convergence of $\hat{\theta}$, some mild conditions on the $L_2$ rates of convergence of the first-step estimators $\hat{\gamma}$ and $\hat{\alpha}$ are required.
\begin{assumption}
\label{Ass:rates of convergence}
For any $k$, $\lVert \hat{\gamma}_{-k} - \gamma_0 \rVert_L^2 = o_P(n^{-\frac{1}{4}}); \lVert \hat{\alpha}_{-k} - \alpha_0 \rVert_L^2 = o_P(1)$. 
\end{assumption} 
Assumption \ref{Ass:rates of convergence} can be satisfied by many flexible nonparametric estimators of $\hat{\gamma}$ including machine learning-based estimators like lasso, random forest, boosting, and neural nets. In practice, these can be useful when the covariate space is large but the true $\tau_0(X)$ has a sparse representation. Calibrated Monte Carlo simulations in \citet{spini_robustness} illustrate this point.
I derive the influence function representation for $\hat{\theta}$ under Assumptions \ref{ass:Unconfoundedness}-\ref{Ass:rates of convergence} which implies the following asymptotic normality result. 
\begin{theorem}[A.N. of $\hat{\theta}$]
\label{thm:Asymptotic normality}
Let Assumptions \ref{ass:Unconfoundedness}--\ref{Ass:rates of convergence} hold. For $\hat{\theta}$ in Equation \eqref{eq:debiasedGMMtheta}:
\begin{align*}
&\sqrt{n} (\hat{\theta} - \theta_0) \overset{d}{\rightarrow} \mathcal{N}(0,S) \ \textrm{with} \ S := (G^{-1}) \Omega (G^{-1})^T, \  \ \\
G &:= \mathbb{E}[D_{\theta} \psi(w,\theta_0,\gamma_0,\alpha_0)], \ \ \Omega := \mathbb{E}[\psi(w,\theta_0,\gamma_0,\alpha_0) \psi(w,\theta_0,\gamma_0,\alpha_0)^T], 
\end{align*}
and $D_{\theta} \psi(\cdot)$ is the Jacobian of the augmented moment condition with respect to $\theta$.
\end{theorem}

\begin{corollary}[A. N. of $\hat{\delta}^*$]
Let $\hat{\delta}^* = -\log(\hat{\nu})$. Then:
\begin{equation*}
\sqrt{n} (\hat{\delta}^* - \delta^*) \overset{d}{\rightarrow} \mathcal{N}\left(0,\frac{S_{11}}{\nu_0^2}\right),
\end{equation*}
where $S_{11}$ is the (1,1) entry of the variance covariance matrix S in Theorem \ref{thm:Asymptotic normality}. 
\end{corollary}
Theorem~\ref{thm:Asymptotic normality} provides point estimates and confidence intervals for $\delta^*$. 
Because $\delta^*$ is defined as the minimal distance required to invalidate the policy-maker’s claim, one would focus on the lower bound: overestimating $\delta^*$ overstates robustness, whereas underestimating is valid but conservative. 
This approach is analogous to \citet{masten2020inference} who report a one-sided confidence region for their breakdown frontier. 

\subsection{Reporting features of the \textit{least-favorable distribution}}
In addition to \(\delta^*\), the researcher may be interested in the 
\textit{least-favorable distribution} \(F_X^*\) itself.
When $\mathcal{X}$ is large, 
representing $F_X^*$ directly can be impractical. 
Instead, the researcher can compare $F_X$ and $F^*_X$ by reporting some moments. Reporting covariate means across treatment arms is standard practice in assessing internal validity (e.g. \citet{rosenbaum1984reducing}). By analogy, comparing moments from $F_X$ and $F_X^*$ could speak to external validity. Here I provide an estimator for a user-specified, finite collection of moments of $F_X^*$.

Theorem \ref{thm:least favorable features} shows the asymptotic properties of the joint estimator for $\theta$ and the additional moments, denoted by $\zeta$. 
\begin{theorem}[De-biased estimator of \textit{least-favorable} moments]
\label{thm:least favorable features}
Let $u: \mathbb{R}^{d} \rightarrow \mathbb{R}^s$, with $u \in (L^{\infty}(\mathcal{X}, P_X))^s$. Let $\zeta_0 = \mathbb{E}_{F^*_X}[u(X)] \in \mathbb{R}^s$. Define the estimating equation for $(\hat{\theta},\hat{\zeta})$, augmenting $\theta$ by the desired moments $\zeta$ of the least-favorable distribution.
\begin{align*}
\hat{\psi}^{u}(\theta,\zeta,\hat{\gamma},\hat{\alpha}) &:= \frac{1}{K} \sum_{k=1}^{K} \frac{1}{|I_k|} \sum_{i \in I_k} \begin{bmatrix} g(W_i,\theta,\hat{\gamma}_{-k}) + \phi(W_i, \theta, \hat{\gamma}_{-k}, \hat{\alpha}_{-k}) \\
g^u(W_i,\theta,\zeta,\gamma_{-k}) + \phi^u(W_i,\theta, \zeta,\hat{\gamma}_{-k},\hat{\alpha}_{-k}) \end{bmatrix} \\
(\hat{\theta},\hat{\zeta}) &:= \arg \min_{(\theta,\zeta) \in \mathbb{R}^{s+2}} \hat{\psi}^{u}(\theta,\zeta,\hat{\gamma},\hat{\alpha})^T \hat{\psi}^{u}(\theta,\zeta,\hat{\gamma},\hat{\alpha}) + o_P(1)
\end{align*}
where $g(\cdot), \phi(\cdot), \gamma(\cdot)$ and $\alpha(\cdot)$ are the same as in Proposition \ref{prop:De-biased GMM nonparamteric influence function}, and $g^u(\cdot)$ and $\phi^u(\cdot)$, taking values in $\mathbb{R}^s$ and defined below.
\begin{align*}
g^u(W_i,\theta,\zeta,\gamma) &= u(X_{i}) \exp(-\lambda(\tau(X_i) -\tilde{\tau}))  - \nu \cdot \zeta \\
\phi^u(W_i,\theta,\zeta,\gamma,\alpha) &= u(X_i) \exp \left( -\lambda (\tau(X_i) - \tilde{\tau} ) \right) (-\lambda) \left( \frac{D_i(Y_i - \gamma_{1}(X_i))}{\pi(X_i)} - \frac{(1-D_i)(Y_i - \gamma_{0}(X_i))}{1-\pi(X_i)} \right)    
\end{align*}
Let Assumptions \ref{ass:Unconfoundedness}--\ref{Ass:rates of convergence} hold. Then: 
\begin{align*}
\frac{1}{\sqrt{n}} \sum_{k=1}^K \sum_{i \in I_k} \psi^u(W_i, \theta, \zeta, \hat{\gamma}_{-k}, \hat{\alpha}_{-k}) &= \frac{1}{\sqrt{n}} \sum_{i=1}^n \psi^u(W_i, \theta, \zeta, \gamma_0, \alpha_0) + o_P(1) \\
\sqrt{n} \left(\begin{bmatrix} \hat{\theta} - \theta_0 \\
\hat{\zeta} - \zeta_0 \end{bmatrix} \right) \overset{d}{\rightarrow} \mathcal{N} (0, S^u) & \ \textrm{with} \ S^u := (G^u)^{-1} \Omega^u
(G^{u'})^{-1}, \\ G^u := \mathbb{E}[D_{\theta,\zeta}  \psi^u(W,\theta,\zeta,\gamma_0,\alpha_0)], \ \
\Omega^u &:= \mathbb{E}[\psi^u(w,\theta_0,\zeta_0,\gamma_0,\alpha_0) \psi^u(w,\theta_0,\zeta_0,\gamma_0,\alpha_0)^T],
\end{align*}
where $D_{\theta,\zeta}$ denotes the Jacobian matrix with respect to the parameters $\theta$ and $\zeta$.
\end{theorem} %
\section{Empirical Application}
\label{section:Empirical Application}
Between March and September 2008, Oregon conducted lottery draws granting winners the option to enroll in Oregon Health Plan (OHP) Standard, a Medicaid expansion program for uninsured residents aged between 19 and 64, who have limited income and assets. In a seminal study, \citet{finkelstein2012oregon} find positive effects of the insurance coverage on a variety of outcomes: health-care utilization (number of prescription, inpatient, outpatient and ER visits), preventive care (cholesterol and diabetes blood test,  mammogram and pap-smear test) and measures of financial strain (outstanding medical debt, denied care, borrow/skip).\footnote{All replicated results are from the publicly available survey data \citep{finkelstein2013oregon}.} Because not all lottery winners exercised their option to enroll, \citet{finkelstein2012oregon} report both an ITT and a LATE estimate. I focus on the ITT as it is likely to be the relevant parameter\footnote{The ITT is just an ATE where the treatment $D$ is the ``the option to enroll in the health insurance": the robustness approach I discussed carries over to the ITT with only notational changes.} for a policy-maker interested in offering the same intervention in another state.

For a policy-maker considering a Medicaid expansion in their state, extrapolating these effects raises an external validity concern: (i) the state's eligible population may differ from the experiment along covariates that shape treatment-effect heterogeneity, and (ii) some of these covariates will be unobserved outside the experiment. Regarding (i), \citet{finkelstein2012oregon} acknowledge substantial demographic differences of the experimental population relative to the US national average: a higher proportion of Oregon's eligibles are White, older and in poorer health. If these covariates are important determinants of treatment effect heterogeneity, the experiment's results may not be robust to the covariate shifts arising when policy is adopted in other states.  Regarding (ii), many survey-specific health covariates, which are likely to be predictive of treatment effect heterogeneity, are exclusively collected in the experiment and have no counterpart in other states. This makes re-weighting procedures like in \citet{hartman_2020_generalizability} and \citet{hsu2020counterfactual} not directly applicable. 
Instead, I propose to quantify the robustness of each outcome studied in \citet{finkelstein2012oregon} by reporting my robustness metric $\delta^*$, which measures how large a covariate shift would need to be to eliminate the effects of the expansion in a new state. To operationalize $\delta^*$, I use the same covariates examined by \citet{finkelstein2012oregon} in their heterogeneity analysis.

\citet{finkelstein2012oregon}'s exploration of treatment effect heterogeneity focuses on six main covariates: binary indicators for the recipient being White, Female, over the age of 50, having more than High School education, living in an Urban area, and being a Smoker (which the authors use as a proxy for health status). In this context, the joint distribution of the covariates is a $64 \times 1$ vector of probabilities, specifying for each combination of the above categories the relative frequency in the sample.\footnote{The full joint distribution needs to specify $2^6=64$ probabilities, summing up to 1.} Although the authors report being underpowered to make inference on heterogeneous treatment effects for each subgroup, overall treatment-effect heterogeneity is still informative for the external validity of aggregate parameters like the ITT in other states.

\begin{table}[!t]
\centering
\footnotesize
\setlength{\tabcolsep}{3pt}
\renewcommand{\arraystretch}{1.15}
\begin{tabular}{|l|c|c|c|c|}
\hline
\textbf{Health-care Utilization} & Prescriptions & Out-patient & ER & In-patient \\
\hline
Experimental ITT &
$\underset{(0.009)}{0.0259}$ &
$\underset{(0.008)}{0.0624}$ &
$\underset{(0.007)}{0.0075}$ &
$\underset{(0.004)}{0.0029}$ \\
$\delta^*(0)$ &
$\underset{(0.046)}{0.054}$ &
$\underset{(0.152)}{0.424}$ &
$\underset{(0.018)}{0.010}$ &
$\underset{(0.009)}{0.003}$ \\
\hline
\textbf{Preventive care} & Cholesterol & Diabetes & Mammogram & Pap test \\
\hline
Experimental ITT &
$\underset{(0.008)}{0.0358}$ &
$\underset{(0.008)}{0.0282}$ &
$\underset{(0.012)}{0.0565}$ &
$\underset{(0.010)}{0.0492}$ \\
$\delta^*(0)$ &
$\underset{(0.073)}{0.158}$ &
$\underset{(0.053)}{0.128}$ &
$\underset{(0.060)}{0.173}$ &
$\underset{(0.047)}{0.173}$ \\
\hline
\textbf{Financial Strain} & Out-of-pocket & Outstanding & Borrow/Skip & Refused care \\
\hline
Experimental ITT &
$\underset{(0.008)}{-0.0579}$ &
$\underset{(0.008)}{-0.0518}$ &
$\underset{(0.008)}{-0.0433}$ &
$\underset{(0.004)}{-0.0104}$ \\
$\delta^*(0)$ &
$\underset{(0.102)}{0.327}$ &
$\underset{(0.103)}{0.310}$ &
$\underset{(0.120)}{0.174}$ &
$\underset{(0.031)}{0.030}$ \\
\hline
\end{tabular}
\caption[Robustness metric for the health-care utilization, preventative care and financial strain outcomes in \citet{finkelstein2012oregon}.]{$\delta^*(0)$ robustness metric for the health-care utilization and financial strain outcomes in \citet{finkelstein2012oregon}. All estimates use survey weights.}
\label{tab:deltas}
\end{table}

Building on this heterogeneity structure, I evaluate the robustness of the aggregate ITT to covariate shifts. Following \citet{finkelstein2012oregon}'s notation, I consider hypotheses of two forms: (i) $ITT_{j} > \tilde{\tau}$ for an increase in healthcare utilization or preventive care outcome $j$;
(ii) $ITT_{j} < \tilde{\tau}$ for a decrease in financial strain outcome $j$. I report calculations for $\tilde{\tau}=0$ to target sign robustness: $\delta^*(0)$ measures the covariate shift required to overturn the policy relevant sign of the ITT (non-negative for utilization and non-positive for financial-strain outcomes). 
Other thresholds can be accommodated by specifying alternative values of $\tilde{\tau}$.\footnote{For example, setting $\tilde{\tau}=t_j:=z_{1-\alpha}\hat{\sigma}_j$ (where $\hat{\sigma}_j$ is the standard error of $\textrm{ITT}_j$ and $z_{1-\alpha}$ is the quantile of the standard normal) targets statistical significance: $\delta^*(t_j)$ measures the covariate shift required for the effect to become statistically non-significant in a one-sided test at level $\alpha$.} For example, $\delta^*(0)=0.327$ in Table~\ref{tab:deltas} is the smallest KL-distance covariate shift compatible with an increase in ``Out-of-pocket expenses". As a descriptive comparison, the result on ``Out-of-pocket expenses" is relatively more robust than the one on ``prescriptions". The next section provides a cardinal interpretation of the magnitude of $\delta^*(0)$. 

\subsection{Benchmarking and interpreting robustness}
\label{subsec:Benchmarking empirical}
Should a particular value of $\delta^*(0)$ in Table \ref{tab:deltas} be considered high or low? To answer this question, I propose a benchmarking exercise based on American Community Survey (ACS) census data. I compare the magnitude of $\delta^*$ to the KL divergence between the covariate distribution of each U.S. state and the experimental distribution, accounting for lottery eligibility. Eligibility for the Oregon lottery required being currently uninsured, having an income below the poverty line, liquid assets below \$2000 and being a state resident. Using  ACS data, I construct a proxy-eligibility rule mirroring these criteria and obtain a proxy-eligible sub-sample in each state. For five out of six heterogeneity covariates, I can construct the same binary indicators as in \citet{finkelstein2012oregon}, they form $X_c$. The experiment's proxy for health status is not available in the ACS and thus forms $X_e$. Pooling years 2008-2012, I obtain each state's covariate distribution, and compute their KL-divergences from the experiment using $X_c$. The results are in Table \ref{tab:kl_pool_by_state}. 
\begin{table}[!t]
\centering
\scriptsize
\setlength{\tabcolsep}{6pt}
\renewcommand{\arraystretch}{1.1}

\begin{minipage}[t]{0.27\linewidth}\vspace{0pt}
\centering
\begin{tabular}{lr}
\hline
\quad \ State & $D_{KL}(G_s\Vert F)$ \\
\hline
\strow{Alabama}{0.327}
\strow{Alaska}{0.974}
\strow{Arizona}{0.217}
\strow{Arkansas}{0.401}
\strow{California}{0.360}
\strow{Colorado}{0.156}
\strow{Connecticut}{0.458}
\strow{Delaware}{0.603}
\strow{DC}{1.23}
\strow{Florida}{0.227}
\strow{Georgia}{0.517}
\strow{Hawaii}{0.738}
\strow{Idaho}{0.320}
\strow{Illinois}{0.382}
\strow{Indiana}{0.164}
\strow{Iowa}{0.225}
\strow{Kansas}{0.234}
\hline
\end{tabular}
\end{minipage}\hspace{0.05\linewidth}
\begin{minipage}[t]{0.3\linewidth}\vspace{0pt}
\centering
\begin{tabular}{lr}
\hline
\quad \ State & $D_{KL}(G_s\Vert F)$ \\
\hline
\strow{Kentucky}{0.510}
\strow{Louisiana}{0.589}
\strow{Maine}{0.537}
\strow{Maryland}{0.596}
\strow{Massachusetts}{0.511}
\strow{Michigan}{0.240}
\strow{Minnesota}{0.205}
\strow{Mississippi}{1.11}
\strow{Missouri}{0.209}
\strow{Montana}{0.788}
\strow{Nebraska}{0.208}
\strow{Nevada}{0.239}
\strow{New Hampshire}{0.293}
\strow{New Jersey}{0.654}
\strow{New Mexico}{0.423}
\strow{New York}{0.410}
\strow{North Carolina}{0.319}
\hline
\end{tabular}
\end{minipage} \hspace{0.05\linewidth}
\begin{minipage}[t]{0.3\linewidth}\vspace{0pt}
\centering
\begin{tabular}{lr}
\hline
\quad \ State & $D_{KL}(G_s\Vert F)$ \\
\hline
\strow{North Dakota}{0.490}
\strow{Ohio}{0.170}
\strow{Oklahoma}{0.288}
\makebox[1.8em][l]{$\mathord{*}\kern0.05em\mathord{\dagger}\kern0.05em\mathord{\bullet}$}\textcolor{teal}{\textbf{Oregon}} & \textcolor{teal}{0.127} \\
\strow{Pennsylvania}{0.266}
\strow{Rhode Island}{0.482}
\strow{South Carolina}{0.429}
\strow{South Dakota}{1.18}
\strow{Tennessee}{0.213}
\strow{Texas}{0.184}
\strow{Utah}{0.190}
\strow{Vermont}{1.61}
\strow{Virginia}{0.276}
\strow{Washington}{0.137}
\strow{West Virginia}{1.09}
\strow{Wisconsin}{0.260}
\strow{Wyoming}{1.54}
\hline
\end{tabular}
\end{minipage}
\caption{KL divergence between each state distribution $G_s$ and the experimental distribution $F$, computed from the observed component in Eq.\eqref{eq:KL_decomposition} using $X_c$. I compare these values to $\delta^*(0)=0.327$ for \emph{Out-of-pocket expenses}. To reason about the unobserved component associated with $X_e$ (health status), I consider hypothetical magnitudes expressed as percentages of the observed KL. Symbols indicate whether the implied total distance remains below $0.327$: $*$ (0\%), $\dagger$ (20\%), and $\bullet$ (100\%). The percentages correspond to no unobserved contribution, a contribution proportional to the relative number of covariates in $X_e$ versus $X_c$, and equal observed and unobserved contributions, respectively.}
\label{tab:kl_pool_by_state}
\end{table}

Consider the \emph{Out-of-pocket expenses} outcome in Table~\ref{tab:deltas} as an example. The ITT estimate indicates that the option to enroll in OHP reduced the probability of out-of-pocket expenses by $5.79\%$. The robustness metric $\delta^*(0)=0.327$ is the smallest shift in the covariate distribution away from the experimental benchmark needed to bring the ITT to zero. Hence for any target state with distribution $G_{X,s}$ if $D_{KL}(G_{X,s} \Vert F_X)<\delta^*(0)$, the experimental claim $(ITT<0)$ must also hold under $G_X$. How does $\delta^*(0)=0.327$ compare to the empirical KL distance benchmarks in Table~\ref{tab:kl_pool_by_state}? As a first pass, suppose that the second term of Eq.~\eqref{eq:KL_decomposition} is zero, so that the reported KL distances based on the census covariates $X_c$ coincide with KL distance over the full $X$. Under this benchmark, Oregon is the closest state to the experimental covariates.\footnote{The two populations are not identical ($KL>0$) due to pooling census years 2008--2012, survey attrition, and potential self-selection into the lottery \citep{finkelstein2012oregon}.} In particular, the Oregon census population has an observable KL distance of $0.127$, so $\delta^*(0)$ is about $2.5$ times the Oregon-to-experiment distance.

Because health status ($X_e$) is not observed in the census, the reported KL distances using only $X_c$ provide a lower bound for $D_{KL}(G_X\Vert F_X)$. The decomposition in Eq.~\eqref{eq:KL_decomposition} helps us reason about the contribution of the missing covariate $X_e$ to KL by expressing it as a percentage of the observed KL based on $X_c$. In Table \ref{tab:kl_pool_by_state}, I consider three natural benchmarks: 0\% (no contribution), 20\%  (contribution proportional to the number of variables in $X_e$ vs $X_c$), and 100\% (equal contribution of observed and unobserved component).\footnote{Using an observable benchmark to conjecture about an unobservable quantity is reminiscent of the approach in \citet{altonji2005selection} \citet{oster2019unobservable} and discussion in \citep{cinelli2020making,masten2022effect,diegert2025axiomatic}.} Twenty-four states remain below the threshold for the 0\% benchmark, nineteen states for the 20\% benchmark, and Oregon, Colorado and Washington meet the very conservative 100\% benchmark.

\begin{figure}[!t]
\centering
\includegraphics[width=0.90\textwidth]{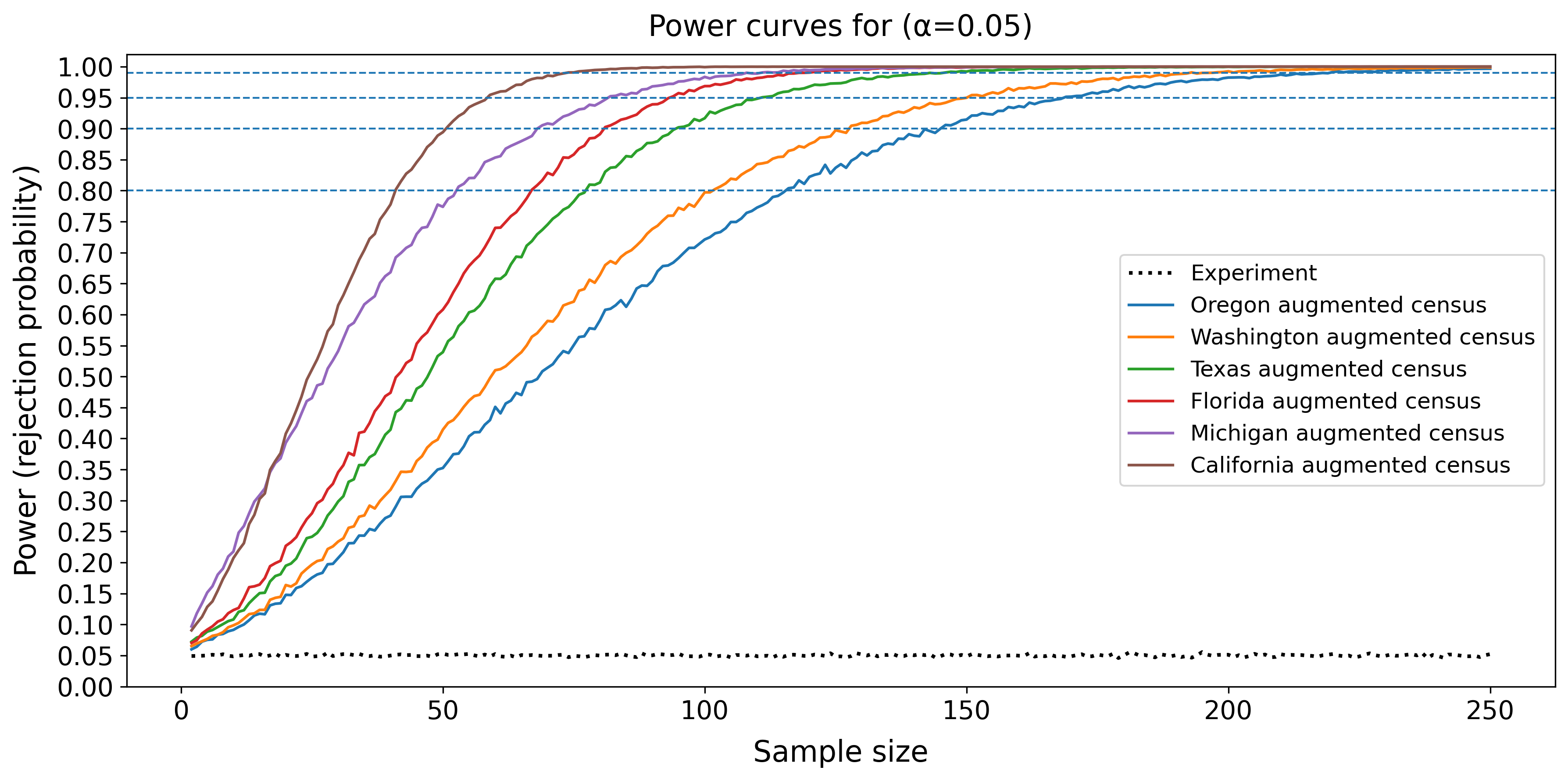}
\caption{Power against sample size. The X-axis shows $m$ for an i.i.d.\ sample drawn from $G_s$ for a few census states. The Y-axis shows the rejection probability for $H_0:G_s=F_X$ against $H_1:G_s\neq F_X$ using the likelihood ratio test at $\alpha=0.05$ (20{,}000 Monte Carlo replications).}
\label{fig:Power curves}
\end{figure}

\subsection{A cardinal interpretation via testing}
I now offer an operational calibration of $\delta^*(0)$ that corroborates in what sense it measures statistical distinguishability. Fix a target environment (state) $s$, and suppose a policymaker could draw an i.i.d.\ sample of size $m$ of covariates from $G_s$. The goal is to test whether a shift from the experimental distribution $F_X$ has occurred using a likelihood ratio (LR) test.
For concreteness, I focus on one example calibrated to the empirical application where both $F$ and $G_s$ are multinomial vectors with 64 cells. We want to test the hypothesis of no covariate shift $H_0:G_s=F$ against $H_1:G_s\neq F$. To do so, we can compare each cell probability ($F_j$) that we would expect under $F$ to the empirical cell probability observed in the sample ($\hat{G}_{s,j}$) and aggregate over all cells.\footnote{Here the empirical distribution $\widehat{G}_s$  is indeed the MLE estimator for $G_s$.} The resulting test statistic can be written as: 
\begin{equation*}
T_{s,m} := 2 m \sum_{j=1}^k \widehat{G}_{s,j} \log\!\left(\frac{\widehat{G}_{s,j}}{F_j}\right)
=2m\,D_{KL}(\widehat G_{s}\|F),
\end{equation*}
Under $H_0$, $T_{s,m}\Rightarrow \chi^2_{k-1}$, so we reject at level $\alpha$ when $T_{s,m}>\chi^2_{k-1,\,1-\alpha}$. Under fixed alternative $G_s\neq F$, $T_{s,m}$ grows linearly in $m$, with its drift governed by the KL distance. Heuristically, larger $D_{KL}(G_s\|F)$ leads to larger $T_{s,m}$ and therefore higher power at a given sample size. 
To show this relationship empirically, I use the census covariate distributions. For each state $s$, I draw $m$ observations from $G_s$, compute the rejection frequency of the LR test at $\alpha=0.05$ (over 20{,}000 Monte Carlo replications), and plot power as a function of $m$. Figure~\ref{fig:Power curves} shows power curves closely tracking the ordering of $D_{KL}(G_s\|F_X)$ across states over a wide range of sample sizes. 

From the simulated power curves, I compute for each state the minimal sample size required to reject $H_0$ at prescribed power benchmarks (0.8, 0.9, 0.95, 0.99) and report them in Table~\ref{tab:Minimal_sample_size}. At each benchmark, states with larger $D_{KL}(G_s\|F_X)$ require fewer observations to detect a shift from $F_X$, consistent with the interpretation of $D_{KL}(G_s\|F_X)$ as a measure of statistical distinguishability from the experiment.
This links directly to the robustness metric. Recall that $\delta^*(0)$ is the smallest KL distance from $F_X$ among covariate shifts that invalidate the experimental claim, attained at $F^*_X$. Define $m^*$ as the minimal sample size needed to detect the shift $F_X$ to $F_X^*$ at a given power benchmark using the LR test. Then $m^*$ provides an operational calibration: any shift that overturns the conclusion must be at least as detectable as $F_X^*$, and therefore should require no more than $m^*$ observations to reach the same power. Equivalently, a shift requiring more than $m^*$ observations at that benchmark cannot overturn the experimental sign conclusion. 
\vskip 0.5cm
\begin{table}[!t]
\centering
\footnotesize
\begin{tabular}{l r cccc}
& & \multicolumn{4}{c}{min $m$: Power} \\
\cmidrule(lr){3-6}
State & $D_{KL}$ & $\geq 80\%$ & $\geq 90\%$ & $\geq 95\%$ & $\geq 99\%$ \\
\midrule
Oregon census      & 0.130 & 116 & 146 & 169 & 220 \\
Washington census  & 0.141 & 102 & 128 & 151 & 193 \\
Texas census       & 0.188 &  78 &  95 & 111 & 146 \\
Florida census     & 0.233 &  67 &  81 &  94 & 118 \\
Michigan census    & 0.246 &  53 &  69 &  82 & 108 \\
California census  & 0.366 &  41 &  51 &  59 &  74 \\
\bottomrule
\end{tabular}

\caption{Effective sample size calculation.}
\label{tab:Minimal_sample_size}
\end{table}
\subsection{Reporting the least-favorable distribution}
To complement the interpretation of $\delta^*$ in Section \ref{subsec:Benchmarking empirical} and to illustrate Theorem~\ref{thm:least favorable features}, I report some of the features of the \textit{least-favorable distribution} $F_X^*$ for an example outcome. Recall that $F_X^*$ is a $64\times 1$ probability vector describing the joint distribution of $X$. I summarize it by reporting the implied marginal proportions of the corresponding indicator variables.
\begin{table}[!t]
\centering
\footnotesize
\begin{tabular}{lccc}
\hline
Variable & $\mu$ (Exp) & $\mu^*$ (LFD) & Mean shift \\
\hline
Age50Plus   & 0.334 & 0.293 & -0.041 \\
Female      & 0.594 & 0.584 & -0.010 \\
MoreThanHS  & 0.333 & 0.279 & -0.054 \\
Smoker      & 0.635 & 0.538 & -0.096 \\
Urban       & 0.749 & 0.765 & +0.017 \\
White       & 0.819 & 0.777 & -0.042 \\
\hline
\end{tabular}
\caption{Marginal means (proportions) under the experimental distribution ($\mu$) and the least-favorable distribution ($\mu^{*}$), and their mean shifts for the outcome ``Out-of-pocket expenses".}
\label{tab:lfd_marginals}
\end{table}
Two remarks are in order. First, because $F_X^*$ specifies the full joint distribution of covariates, the marginals in Table \ref{tab:lfd_marginals} provide a coarse but informative summary of the direction of the implied shift. Second, these mean shifts are closely linked to treatment-effect heterogeneity. The least-favorable distribution shifts weight away from covariate groups (e.g., smokers, older, more educated, and White recipients) for which the conditional ITT is more negative, and toward groups for which the conditional ITT is positive. This reweighting is how the overall ITT can be driven to the threshold $\tilde{\tau}=0$. The largest shift occurs along the Smoker indicator, which, consistent with \citet{finkelstein2012oregon}'s use of it as a proxy for health, is strongly predictive of health-related expenses. This is precisely where reporting $\delta^*(0)$ and the benchmarking exercise are most useful: when a heterogeneity-predictive covariate has no census counterpart and threatens external validity, we can assess the magnitude of the minimal invalidating shift relative to observed shifts in covariates that do have census counterparts in plausible implementation environments.%
\section{Conclusion}
\label{section:Conclusion}
To measure the external validity of experimental claims of the form $(ATE>\tilde{\tau})$ under covariate shifts, I propose a robustness metric $\delta^*(\tilde{\tau})$, defined as the distance to the nearest covariate distribution that overturns the claim. Applied researchers can report $\delta^*(\tilde{\tau})$ as a ``third number'' alongside the point estimate and standard error, and I provide interpretation and calibration exercises that make its magnitude operational. These include benchmarking $\delta^*$ against empirically observed covariate shifts across plausible implementation environments and mapping $\delta^*$ into the sample size required to detect a shift at a given power benchmark. 

Extending the framework to other linear policy parameters is straightforward: just like for the ATE, the closed-form solution yields a simple de-biased GMM procedure that accommodates estimation of the CATE via flexible ML methods (e.g., LASSO, random forests, neural networks). By contrast, for nonlinear distributional parameters, a comparable closed-form characterization and the development of estimation and inference remains an interesting open problem.

\begin{appendix}
\section{Appendix}
\subsection{Proof of Theorem \ref{thm:closed form solution}}
The proof uses the variational characterization of \citet{donsker1975asymptotic}. \begin{lemma}[\citet{donsker1975asymptotic}]
\label{lem:KL_decomposition}
Let $F_X^*$ satisfy $\frac{d F_X^*}{dF_X} = \frac{\exp(-\lambda(\tau(x) - \tilde{\tau}))}{\int_{\mathcal{X}} \exp(-\lambda(\tau(x) - \tilde{\tau})) dF_X} $. For any probability measure $\tilde{F}_X$ such that $\tilde{F}_X \ll F_X$, we have:
\begin{equation*}
\log \left(\int_{\mathcal{X}} \exp(-\lambda(\tau(x) - \tilde{\tau})) d F_X \right) = - \left[ \int_\mathcal{X} \lambda (\tau(x)- \tilde{\tau}) d\tilde{F}_X(x)  + D_{KL}(\tilde{F}_X||F_X) \right] + D_{KL}(\tilde{F}_X||F_X^*) 
\end{equation*}
\end{lemma}

We can now prove Theorem \ref{thm:closed form solution}. By the Radon-Nikodym theorem, $\frac{d F_X'}{d F_X}$ exists and $\textrm{supp}\left(\frac{d F_X'}{d F_X}\right) \subset \mathcal{X}$. Then, the optimization problem in Eq.\ref{eq:Inf_Problem}-\ref{eq:Constraint} is equivalent to:
\begin{align*}
\inf_{F_X': \ P_X' \ll P_X} & D_{KL}(F_X'||F_X)  \\
s.t. \ & \int_{\mathcal{X}} \tau(x) \frac{d F_X'}{d F_X} dF_X(x) = \tilde{\tau}; \ \quad P_X'(\mathcal{X}) = 1.
\end{align*}

\begin{proof}
i) From Lemma \ref{lem:KL_decomposition} we have:
\begin{equation*}
\log \left( \int_{\mathcal{X}} \exp (-\lambda(\tau(x) - \tilde{\tau})) d F_X \right) = D_{KL}(\tilde{F}_X||F_X^*) - D_{KL}(\tilde{F}_X||F_X) - \int_{\mathcal{X}} \lambda(\tau(x) - \tilde{\tau}) d \tilde{F}_X.
\end{equation*}
Since the term $\log \left( \int_{\mathcal{X}} \exp (-\lambda(\tau(x) - \tilde{\tau}) d F_X \right)$ does not depend on $\tilde{F}_X$ we must have:
\begin{align*}
\arg \min_{\tilde{F}_X \ll F_X} D_{KL}(\tilde{F}_X||F_X^*) &= \arg \max_{\tilde{F}_X \ll F_X} - \int_{X} \lambda (\tau(x) - \tilde{\tau}) d \tilde{F}_X - D_{KL}(\tilde{F}_X||F_X) \\
&= \arg \min_{\tilde{F}_X \ll F_X} \int_{X} \lambda (\tau(x) - \tilde{\tau}) d \tilde{F}_X + D_{KL}(\tilde{F}_X||F_X),
\end{align*}
but clearly $F_X^* = \arg \min_{\tilde{F}_X \ll F_X} D_{KL}(\tilde{F}_X||F_X^*)$ so we must have 
\begin{align*}
F_X^* = \arg \min_{\tilde{F}_X \ll F_X} D_{KL}(\tilde{F}_X||F_X) + \lambda \int_{X} (\tau(x) - \tilde{\tau}) d \tilde{F}_X.    
\end{align*}
ii) Because $D_{KL}(F_X^*||F_X^*) = 0$ the value of the minimization problem:
\begin{align*}
& \ \ \ \min_{\tilde{F}_X \ll F_X} D_{KL}(\tilde{F}_X||F_X) + \lambda \int_{X} (\tau(x) - \tilde{\tau}) d \tilde{F}_X \\
&= \min_{\tilde{F}_X \ll F_X} D_{KL}(\tilde{F}_X|| F_X^*) - \log \left( \int_{\mathcal{X}} \exp (-\lambda(\tau(x) - \tilde{\tau})) d F_X \right) \\
&= -\log \left( \int_{\mathcal{X}} \exp (-\lambda(\tau(x) - \tilde{\tau})) d F_X \right).
\end{align*}
\end{proof}
\subsection{Auxiliary lemmas}
\noindent Proving Theorem \ref{thm:Asymptotic normality} requires some lemmas. Their derivations are in \citet{spini_robustness}.
\begin{lemma}[\citet{kennedy2020sharp}] 
\label{lem:Kennedy} Let $\hat{g}(\cdot)$ be a function estimated from the $I_k^c$ sample and evaluated on the $I_k$ sample. Then $(\mathbb{P}_n - \mathbb{P}) (\hat{g} - g_0) = O_P \left( \frac{\lvert \hat{g} - g_0 \rvert }{\sqrt{n}} \right)$. 
\end{lemma}

\begin{lemma}
\label{lem:Frechet differentiability}
For $\psi(\cdot)$ in Equation \eqref{eq:De-biased GMM} and $\bar{\psi}(\theta,\gamma,\alpha) = \mathbb{E}[\psi(w,\theta,\gamma,\alpha)]$ we have: 
\begin{enumerate}
\item $\bar{\psi}(\theta_0,\gamma,\alpha_0)$ is twice continuously Fr\'echet differentiable in a neighborhood of $\gamma_0$. 
\item If $\Lambda$ is bounded then $\forall \theta \in \Theta$, $\bar{\psi}(\theta,\gamma,\alpha_0) \leq \bar{C} \lVert \gamma - \gamma_0 \rVert^2_{L_2}$.
\end{enumerate}
\end{lemma}

\begin{lemma}[Jacobian consistency]
\label{lem:Jacobian Consistency}
For the Jacobian G defined as:
\begin{equation*}
G = \mathbb{E}[D \psi(w,\theta_0,\gamma_0,\alpha_0)] = \mathbb{E} \left[\frac{\partial}{\partial \theta} \psi(w,\theta_0,\gamma_0,\alpha_0) \right]    
\end{equation*}
and $\hat{\theta} \overset{p}{\rightarrow} \theta_0$ we have $\lVert \frac{\partial \hat{\psi}(\hat{\theta})}{\partial \theta} - G \rVert = o_P(1)$. If $G^{-1}$ exists then $\lVert \hat{G}^{-1} - G^{-1}  \rVert= o_P(1)$. 
\end{lemma}

\begin{lemma}[$\sqrt{n}$ - consistency]
\label{lem:square root consistency}
Let Assumption \ref{Ass:rates of convergence} hold. Then
\begin{equation*}
\frac{1}{\sqrt{n}} \sum_{k=1}^K \sum_{i \in I_k} g(W_i,\theta,\hat{\gamma}_{-k}) + \phi(W_i,\tilde{\theta}_{-k},\hat{\gamma}_{-k},\hat{\alpha}_{-k}) = \frac{1}{\sqrt{n}} \sum_{i=1}^n \psi(W_i,\theta,\gamma_0,\alpha_0) + o_P(1) \end{equation*}
\end{lemma}

\subsection{Proof of Theorem \ref{thm:Asymptotic normality}}
Denote $\hat{G} = \frac{\partial \hat{\psi}(w,\hat{\theta},\hat{\gamma})}{\partial \theta}$. First note that by Lemma \ref{lem:Jacobian Consistency} we have $\lVert \hat{G}^{-1} - G^{-1} \rVert = o_P(1)$. 
Now by the central limit theorem and Lemma \ref{lem:square root consistency} we have:
\begin{equation*}
\frac{1}{|K|} \sum_{k \in K} \Big(\frac{1}{\sqrt{n}} \sum_{i \in I_k} g(W_i,\theta,\gamma_0) + \phi(W_i,\tilde{\theta}_{-k},\hat{\gamma}_{-k},\hat{\alpha}_{-k}) \Big) 
\overset{d}{\rightarrow}\mathcal{N}(0,\Omega)
\end{equation*}
Finally a standard GMM Taylor linearization gives the desired result:
\begin{align*}
\sqrt{n} \begin{bmatrix} \nu - \nu_0 \\
\lambda - \lambda_0
\end{bmatrix} =& \Bigg \{ \frac{\partial}{\partial \theta} \hat{\psi}(w,\theta_0,\hat{\gamma},\hat{\alpha})' V \frac{\partial}{\partial \theta} \hat{\psi}(w,\theta_0,\hat{\gamma},\hat{\alpha}) \Bigg\}^{-1} \frac{\partial}{\partial \theta} \hat{\psi}(w,\theta_0,\hat{\gamma},\hat{\alpha})' V \\
\times & \frac{1}{|K|} \sum_{k \in K} \Big(\frac{1}{\sqrt{n}} \sum_{i \in I_k} g(W_i,\theta,\hat{\gamma}_{-k}) + \phi(W_i,\tilde{\theta}_{-k},\hat{\gamma}_{-k}) \Big) \\
&= (G'VG)^{-1} G'V \left( \frac{1}{|K|} \sum_{k \in K} \frac{1}{\sqrt{n}} \sum_{i \in I_k} \psi(W_i,\theta,\gamma_0,\alpha_0) \right) + o_P(1) \overset{d}{\rightarrow} \mathcal{N}(0,S)
\end{align*}

\subsection{Proof of theorem \ref{thm:least favorable features}}

\noindent The proof of Theorem \ref{thm:least favorable features} follows the same structure of Theorem \ref{thm:Asymptotic normality} and is omitted.
\end{appendix}

\bibliography{bibliography_main}

\end{document}